\documentclass[letterpaper, 10 pt, conference]{ieeeconf}  

\IEEEoverridecommandlockouts                              

\usepackage{titlesec}
\usepackage{booktabs}
\usepackage{caption-light}
\usepackage{nicefrac}
\usepackage{bm}
\usepackage{amsmath}
\usepackage{amssymb}
\usepackage{tikz}
\usepackage{pgfplots}
\pgfplotsset{compat=1.18} 
\usetikzlibrary{positioning}
\makeatletter
\let\NAT@parse\undefined
\makeatother

\usepackage{hyperref}
\hypersetup{
  hidelinks,
  pdfborder={0 0 0},
}

\usepackage{ifthen}

\usepackage{amsthm}

\renewcommand{\theparagraph}{\alph{paragraph}}

\titleformat{\paragraph}
  {\normalfont\normalsize\bfseries} 
  {\theparagraph)}                  
  {0.5em}                           
  {}                                

\titlespacing*{\paragraph}
  {0pt}                             
  {3.25ex plus 1ex minus .2ex}      
  {0em}                             

\titleformat{\subparagraph}[runin]
  {\normalfont\normalsize\bfseries}
  {}
  {0pt}
  {}

\titlespacing{\subparagraph}
  {0pt}
  {\medskipamount}
  {1em}

\newtheorem{theorem}{Theorem}

\newtheorem{remark}{Remark}

\usepackage{xcolor}

\usepackage{mathtools}
\usepackage{amssymb}
\usepackage{ifthen}
\usepackage{algorithm}

\newcommand{\define}{=\vcentcolon}
\newcommand{\defined}{\vcentcolon=}

\newcommand{\reals}{\mathbb{R}}

\newcommand{\Mat}[1][]{\ifthenelse{\equal{#1}{}}{\textnormal{Mat}}{\textnormal{Mat}(#1)}}

\newcommand{\tangent}[1]{
    \ifthenelse{\equal{#1}{}}
    {{T}}
    {{T_{#1}}}
}
\newcommand{\dualtangent}[1]{
    \ifthenelse{\equal{#1}{}}
    {{T^*}}
    {{T_{#1}^*}}
}

\newcommand{\Adjoint}[1]{
    \ifthenelse{\equal{#1}{}}
    {\textnormal{Ad}}
    {\textnormal{Ad}_{#1}}
}

\newcommand{\adjoint}[1]{
    \ifthenelse{\equal{#1}{}}
    {\textnormal{ad}}
    {\textnormal{ad}_{#1}}
}
\newcommand{\coAdjoint}[1]{
    \ifthenelse{\equal{#1}{}}
    {\textnormal{Ad}^*}
    {\textnormal{Ad}^*_{#1}}
}

\newcommand{\coadjoint}[1]{
    \ifthenelse{\equal{#1}{}}
    {\textnormal{ad}^*}
    {\textnormal{ad}^*_{#1}}
}

\newcommand{\set}[1]{\left\{#1\right\}}

\newcommand{\given}{\,\middle|\,}

\newcommand{\abs}[1]{\left|#1\right|}

\newcommand{\inverse}{{-1}}

\newcommand{\transpose}{\intercal}

\newcommand{\diag}[1]{\operatorname{diag}\left\{#1\right\}}

\newcommand{\Vector}[1]{{#1}}
\newcommand{\Matrix}[1]{\mathrm{#1}}
\newcommand{\identity}[1]{
    \ifthenelse{\equal{#1}{}}
    {\Matrix{I}}
    {\Matrix{I}_{#1}}
}

\newcommand{\indicator}[1]{\textbf{1}_{\left\{#1\right\}}}
\NewDocumentCommand{\Probability}{o m}{
  \mathsf{P}
  \IfValueT{#1}{_{#1}}
  \IfNoValueT{#1}{}
  \left[#2\right]
}
\NewDocumentCommand{\Expectation}{o m}{
  \mathsf{E}
  \IfValueT{#1}{_{#1}}
  \IfNoValueT{#1}{}
  \left[#2\right]
}
\NewDocumentCommand{\Variance}{o m}{
  \mathsf{Var}
  \IfValueT{#1}{_{#1}}
  \IfNoValueT{#1}{}
  \left[#2\right]
}
\NewDocumentCommand{\Covariance}{o m}{
  \mathsf{Cov}
  \IfValueT{#1}{_{#1}}
  \IfNoValueT{#1}{}
  \left[#2\right]
}

\newcommand{\Order}[1]{O\left(#1\right)}

\newcommand{\Gaussian}{\mathcal{N}}

\newcommand{\onehot}[2][]{\mathsf{e}_{#2}\ifthenelse{\equal{#1}{}}{}{^{(#1)}}}

\newcommand{\BDE}{\text{B}\Delta\text{E}}
\newcommand{\timeSpace}{\mathcal{T}}
\newcommand{\initial}{0}
\newcommand{\horizon}{T}
\newcommand{\order}{\tau}
\newcommand{\stateDimension}{\mathsf{d}}
\newcommand{\observationDimension}{\mathsf{m}}
\newcommand{\stateProcess}{X}
\newcommand{\observationProcess}{Z}
\newcommand{\transitionModel}{\Matrix{A}}
\newcommand{\observationModel}{\Matrix{C}}
\newcommand{\processNoise}{B}
\newcommand{\observationNoise}{W}
\newcommand{\covarianceProcessNoise}{\Matrix{Q}}
\newcommand{\covarianceObservationNoise}{\Matrix{R}}
\newcommand{\estimator}{S}
\newcommand{\estimated}[1]{\hat{#1}}
\newcommand{\covarianceError}{P}
\newcommand{\KalmanGain}{L}
\newcommand{\adjusted}[1]{\underline{#1}}
\newcommand{\dualstate}{y}
\newcommand{\control}{u}
\newcommand{\momentum}{p}

\newcommand{\mean}[1]{\bar{#1}}
\newcommand{\covariance}{\Sigma}
\newcommand{\initialMean}{\mu_{\initial}}
\newcommand{\initialCovariance}{\covariance_{\initial}}
\newcommand{\controlSpace}{\mathcal{U}}
\newcommand{\dualCost}{\mathsf{J}}
\newcommand{\iteration}{l}
\newcommand{\optimum}{{*}}

\newcommand{\covState}{\Sigma_{\stateProcess\stateProcess}}
\newcommand{\covStateObs}{\Sigma_{\stateProcess\observationProcess}}

\newcommand{\covObs}{\Sigma_{\observationProcess\observationProcess}}
\newcommand{\weight}{\Matrix{K}}
\newcommand{\smoothingWeight}{\weight^{(\text{s})}}
\newcommand{\filteringWeight}{\weight^{(\text{f})}}
\newcommand{\bias}{\Vector{b}}
\newcommand{\smoothingBias}{\bias^{(\text{s})}}
\newcommand{\filteringBias}{\bias^{(\text{f})}}
\newcommand{\innovations}{\Matrix{D}}
\newcommand{\stepsize}{\eta}
\newcommand{\dummyMatrix}{\Matrix{M}}

\newcommand{\trackingRate}{\alpha}
\newcommand{\oscillationAngle}{{\raisebox{0.1ex}{$\vartriangle$}\theta}}
\newcommand{\fractionCoefficient}{a}
\newcommand{\polynomialOrder}{q}
\newcommand{\discreteFrequency}{\Omega}
\newcommand{\error}{\text{MSE}}
\usepackage{xfp}

\usepackage{etoolbox}

\usepackage{pgfplotstable}

\def\lineWidth{1.2pt}    

\definecolor{colorBlue}{RGB}{31,119,180}
\definecolor{colorOrange}{RGB}{255,127,14}
\definecolor{colorGreen}{RGB}{44,160,44}
\definecolor{colorRed}{RGB}{214,39,40}

\def\kalmanColor{colorBlue}       
\def\smoothingColor{colorGreen}     
\def\causalColor{colorOrange}      
\def\dualColor{colorRed}      

\def\bS{\mathbb{S}}

\def\Sec#1{Sec.~\ref{#1}}
\def\Fig#1{Fig.~\ref{#1}}

\def\notes#1{\marginpar{\tiny #1}\typeout{Notes!
Notes!
Notes!
}}
\renewcommand{\notes}[1]{\typeout{notes!}}

\def\Sec#1{Sec.~\ref{#1}}

\def\Sec#1{Sec~\ref{#1}}

\def\Fig#1{Fig.~\ref{#1}}
\def\Sec#1{Sec.~\ref{#1}}

\def\beq{\begin{eqnarray}} 
\def\bc{\begin{center}} 
\def\be{\begin{enumerate}}
\def\bi{\begin{itemize}} 
\def\bs{\begin{small}}
\def\bS{\begin{slide}}
\def\ec{\end{center}} 
\def\ee{\end{enumerate}}
\def\ei{\end{itemize}}
\def\es{\end{small}}
\def\eS{\end{slide}}
\def\eeq{\end{eqnarray}}

\newcommand{\newP}[1]{\medskip\noindent{\bf #1:}}

\def\Sec#1{Sec.~\ref{#1}}
\def\Thm#1{Thm.~\ref{#1}}

\title{\LARGE \bf
Duality Theory for Non-Markovian Linear Gaussian Models
}

\author{
  Aditya Kudre$^{1,2}$, Heng-Sheng Chang$^{1,3}$, and Prashant G. Mehta$^{1,3}$
  \thanks{This work is supported in part by the AFOSR award FA9550-23-1-0060, the NSF award 2336137.}
  \thanks{%
    $^{1}$Coordinated Science Laboratory,
    $^{2}$Department of Electrical and Computer Engineering, and
    $^{3}$Department of Mechanical Science and Engineering at the  University of Illinois Urbana-Champaign. 
    Corresponding e-mail: hschang@illinois.edu
  }%
}

\begin{document}

\bstctlcite{BSTcontrol}

\maketitle
\thispagestyle{empty}
\pagestyle{empty}

\begin{abstract}
  This work develops a duality theory for partially observed linear Gaussian models in discrete time.
  The state process evolves according to a causal but non-Markovian (or higher-order Gauss-Markov) structure, captured by a lower-triangular transition operator, which is related to transformer, with $\horizon$ as the context length.
  The main contributions are:
  (i) a dual control system for the linear Gaussian model, formulated as a backward difference equation ($\BDE$);
  (ii) a duality principle establishing that a specific linear-quadratic optimal control problem for the $\BDE$ is dual to the filtering problem for the partially observed model;
  and (iii) an explicit optimal control formula yielding a novel (transformer-like) linear predictor, referred to as the dual filter, whose computational complexity scales linearly in the time horizon $\horizon$, in contrast to the $\Order{\horizon^3}$ cost of classical smoothing and Wiener-Hopf approaches.
\end{abstract}

\section{Introduction}
\label{sec:intro}

A decoder-only transformer is an inference architecture for the next-token prediction
as follows: 
\begin{equation*}
  \set{\observationProcess_0, \observationProcess_1, \hdots, \observationProcess_{\horizon - 1}} \mapsto [\text{prediction of}\;Z_T],
\end{equation*}
where the random variable (token) at time $t$, $\observationProcess_t$, is an element of a discrete finite set (vocabulary) and the prediction is the conditional probability vector $\Probability{\observationProcess_{\horizon}\given \observationProcess_{0:\horizon-1}}$.
Pertinent to the present paper, there are three salient features of the transformer algorithm:
\begin{enumerate}
  \item Non-recursive structure: The algorithm processes the entire sequence as a batch rather than recursively updating a state as a function of time.  
  \item Embedding: Discrete-valued tokens are mapped into a continuous $\stateDimension$-dimensional vector space, denoted as $\rho_t \in \reals^\stateDimension$, where all subsequent operations are conducted.
  \item Layer transformation: Each layer in a decoder-only transformer implements a causal transformation as follows:
\begin{equation*}  
  \hspace*{-\leftmargin}\begin{bmatrix} \rho_0 & \rho_1 & \cdots & \rho_{\horizon-1} \end{bmatrix}_{\stateDimension \times \horizon} \mapsto \begin{bmatrix} \rho_0^+ & \rho_1^+ & \cdots & \rho_{\horizon-1}^+ \end{bmatrix}_{\stateDimension \times \horizon},
\end{equation*}
where causal means $\rho_t^+$ depends upon $\{\rho_0,\rho_1,\hdots,\rho_t\}$. 
\end{enumerate}
See~\cite{phuong2022formal} for additional details.

Inspired by the decoder-only transformer, we re-visit inference architectures for linear prediction of vector-valued Gaussian processes---the most celebrated of which are the Kalman and Wiener filters.  Specifically, in this paper, $\observationProcess_t$ is $\reals^\observationDimension$-valued and the particular choice of prediction is the conditional expectation vector $\Expectation{\observationProcess_{\horizon}\given \observationProcess_{\initial:\horizon-1}} \define \estimated{\observationProcess}_{\horizon|\horizon-1}$---which is a natural analogue for continuous-valued Gaussian processes.  To help relate our work to transformers, we assume the following structure for the linear Gaussian model:
\begin{enumerate}
\item The observation process is modeled as $\observationProcess_t = \observationModel_t \stateProcess_t + \observationNoise_t$ , where $\stateProcess_t$ denotes an $\reals^\stateDimension$-valued hidden state process and $\{\observationNoise_t: 0 \le t \le \horizon\}$ is an independent Gaussian noise.  The matrix $C_t$ plays the role of the embedding in the sense that it maps the hidden state to the observation space. 
\item The model of the hidden state process $\{\stateProcess_t: 0 \leq t \leq \horizon\}$ is causal but non-Markovian (or higher-order Markovian): At time $t$, the hidden state $\stateProcess_t$ is allowed to depend upon the entire history $\{\stateProcess_s: 0 \leq s \leq t-1\}$.  The modeling choice is inspired by the long range dependencies that transformer architectures are designed to capture.
\end{enumerate}
The probabilistic graphical model is depicted in~\Fig{fig:model}.  The model can be viewed as a linear Gaussian abstraction of the sequential structure underlying decoder-only transformers.

The main difficulty comes from the non-Markovian structure of the hidden state 
process. Specifically, a recursive Kalman filter-type algorithm will suffer from 
a growing state-dimension as time progresses, because the filter must maintain 
an estimate of the full joint state $(X_0, X_1, \ldots, X_t)$ as a sufficient 
statistic. For this reason, non-recursive algorithms become an attractive 
alternative, and duality theory provides a principled route to their derivation.


\newP{Original Contributions} 
This paper extends the classical Kalman duality theory to non-Markovian 
linear Gaussian models, formally introduced in \Sec{sec:problem}.  
Specifically, the following questions are of interest:
\begin{description}
  \item \textbf{Q1.} What is the dual control system?
    \smallskip
  \item \textbf{Q2.} What is the optimal control problem that is dual to the optimal filtering problem?
  \smallskip
  \item \textbf{Q3.} How does the solution of the dual optimal control problem relate structurally to transformer layer operations?
\end{description}
In this paper, we address each of these questions (see \Sec{sec:dualcontrolsys} for Q1, \Sec{sec:dual-filter} for Q2, and \Sec{sec:transformer-arch-relation} for Q3).  
The solution to the dual optimal control problem yields an explicit optimal control law. 
This in turn leads to a novel linear predictor --- the dual filter --- whose connections to classical methods and to transformer architectures are discussed in \Sec{sec:classical-methods} and \Fig{fig:transformer-objective}, 
respectively.

\newP{Outline} 
The remainder of this paper is organized as follows. The linear Gaussian model and problem formulation are introduced in \Sec{sec:problem}. 
Duality theory is developed in \Sec{sec:duality} and the main results are presented in \Sec{sec:results}. 
Classical methods are compared in \Sec{sec:classical-methods}. 
Numerical experiments are reported in \Sec{sec:numerics} and conclusions are given in \Sec{sec:conclusion}. 
Proofs and technical details are collected in the appendices.

\begin{figure*}[t]
  \medskip
  \centering
  \begin{tikzpicture}[scale=0.9, every node/.style={scale=0.8}, font=\footnotesize]
    \input{figures/model.tex}
  \end{tikzpicture}
  \caption{
    Graphical representation of the non-Markovian linear Gaussian model with order $\order$ and horizon $\horizon$. 
    The state process $\stateProcess$ is represented by circles and the observation process $\observationProcess$ is represented by squares. 
    The arrows represent the causal dependencies between the processes, including transition matrices $\transitionModel_{t,s}$ and $\observationModel_t$.
    The dashed arrows indicate the prediction task of estimating $\observationProcess_\horizon$ given the past observations $\observationProcess_{\initial:\horizon-1}$. 
  }
  \label{fig:model}
\end{figure*}

\section{Problem Formulation}
\label{sec:problem}

We introduce the non-Markovian linear Gaussian model.  
The two processes are as follows:
\begin{align*}
  \text{(hidden state process)} &\quad \stateProcess\defined\set{\stateProcess_\initial,\stateProcess_1,\hdots,\stateProcess_{\horizon-1},\stateProcess_\horizon},\\
  \text{(observation process)} &\quad \observationProcess\defined\set{\observationProcess_0,\observationProcess_1,\hdots,\observationProcess_{\horizon-1},\observationProcess_\horizon},
\end{align*}
where $\observationProcess_\horizon$ is to be predicted rather than observed. 
Time is indexed by $t \in \set{0,1,\hdots,\horizon}$ with finite horizon $\horizon$. 
The state $\stateProcess_t \in \reals^{\stateDimension}$ and observation $\observationProcess_t \in \reals^{\observationDimension}$, where $\stateDimension$ and $\observationDimension$ are finite.
A key feature of the model is its causal, non-Markovian structure, illustrated in~\Fig{fig:model}, and defined as follows:
\begin{subequations}
  \begin{align}
    \stateProcess_{t} &= \sum_{s=1}^{\min(\order,t)} \transitionModel_{t,s}\stateProcess_{t-s} + \processNoise_{t}\label{eq:process},\quad 1\leq t \leq \horizon \\
    \observationProcess_t &= \observationModel_t \stateProcess_t + \observationNoise_t ,\quad 0\leq t \leq \horizon\label{eq:observation}
  \end{align}%
  \label{eq:model}%
\end{subequations}
where $\order$ is the deterministic model-order parameter.
The observation matrices $\observationModel_t$ and transition matrices $\transitionModel_{t,s}$ are deterministic and known for all $t,s$.  
Stochasticity is introduced through three mutually independent Gaussian sources: 
(1) the initial condition $\stateProcess_\initial \sim \Gaussian(\initialMean, \initialCovariance)$ with mean $\initialMean$ and covariance $\initialCovariance\succeq 0$, 
(2) the white Gaussian noise (WGN) process $\processNoise_t \sim \Gaussian(0, \covarianceProcessNoise_t)$ with $\covarianceProcessNoise_t\succeq 0$, and 
(3) the WGN process $\observationNoise_t \sim \Gaussian(0, \covarianceObservationNoise_t)$ with $\covarianceObservationNoise_t\succ 0$. 

Denote
$
  \timeSpace\defined\set{1, \hdots, \horizon}$ and $\timeSpace_-\defined\set{0, \hdots, \horizon-1}$. 
A sequence $\set{V_t:0\leq t\leq \horizon}$ is denoted as a column vector 
\begin{equation*}
  V_{t_1:t_2}\defined\begin{bmatrix}
    V_{t_1} \hfill\\
    V_{t_1+1} \\
    \vdots \\
    V_{t_2}\hfill
  \end{bmatrix},\quad
  V_{\timeSpace}\defined V_{1:\horizon},\quad
  V_{\timeSpace_-}\defined V_{\initial:\horizon-1}.
\end{equation*}
Using this notation, the model \eqref{eq:model} is compactly expressed as follows:
\begin{subequations}\label{eq:compact}
  \begin{align}
    \stateProcess_{\timeSpace} &= \transitionModel \stateProcess_{\timeSpace_-} + \processNoise_{\timeSpace} , \qquad X_0 \sim  \Gaussian(\initialMean, \initialCovariance),\quad \label{eq:compact-process}\\
    \observationProcess_{\timeSpace_-} &= \observationModel \stateProcess_{\timeSpace_-} + \observationNoise_{\timeSpace_-},\label{eq:compact-observation}
  \end{align}
\end{subequations}
and $\observationProcess_\horizon = \observationModel_\horizon \stateProcess_\horizon + \observationNoise_\horizon$ is stated separately as it is the variable to be predicted.
The block matrices $\transitionModel$ and $\observationModel$ are defined as follows:
\begin{align*}
  \transitionModel &\defined \begin{bmatrix}
    \transitionModel_{1,1} & 0 & \cdots & 0 \\
    \transitionModel_{2,2} & \transitionModel_{2,1}  & \cdots & 0 \\
    \vdots & \vdots & \ddots & \vdots \\
    \transitionModel_{\horizon,\horizon} & \transitionModel_{\horizon,\horizon-1} & \cdots & \transitionModel_{\horizon,1}
  \end{bmatrix},\\[2pt]
    \observationModel &\defined \diag{\observationModel_{0}, \dots, \observationModel_{\horizon-1}}.
\end{align*}
Similarly, the block-diagonal covariance matrices for the two WGNs are denoted by
\begin{align*}
  \covarianceProcessNoise &\defined \diag{\covarianceProcessNoise_1, \dots, \covarianceProcessNoise_{\horizon}}, \\
  \covarianceObservationNoise &\defined \diag{\covarianceObservationNoise_0, \dots, \covarianceObservationNoise_{\horizon-1}}.
\end{align*}
Because of the lower-triangular form of the matrix $\transitionModel$, the model~\eqref{eq:compact} is well-posed and the pair $(\stateProcess,\observationProcess)$ is well-defined (i.e., $\stateProcess_t$ is uniquely determined by $\stateProcess_\initial$ and the noises).

\subsection{Problem Statement}
Our goal is to compute
\begin{equation}
  \text{(next-step prediction)}\quad\estimated{\observationProcess}_{\horizon|\horizon-1} \defined \Expectation{\observationProcess_{\horizon} \given \observationProcess_{\initial:\horizon-1}}.
  \label{eq:estimation-problem}
\end{equation}
Using~\eqref{eq:observation}, $\estimated{\observationProcess}_{\horizon|\horizon-1} = \observationModel_{\horizon} \estimated{\stateProcess}_{\horizon|\horizon-1}$, where the state estimate $\estimated{\stateProcess}_{\horizon|\horizon-1} \defined \Expectation{\stateProcess_{\horizon} \given \observationProcess_{\initial:\horizon-1}}$. 

For any deterministic vector $f \in \reals^{\stateDimension}$, consider
\begin{equation}
f^\transpose \estimated{\stateProcess}_{\horizon|\horizon-1} = \text{(constant)}  - \sum_{t=\initial}^{\horizon-1} \control^\transpose_{t} \observationProcess_{t}, \label{eq:predictor}
\end{equation}
where $\control \defined \{\control_\initial, \dots, \control_{\horizon-1}\}$ denotes a deterministic sequence of $\reals^\observationDimension$-valued weights, hereafter referred to as the control. By setting $f^\transpose$ to the rows of the observation matrix $\observationModel_{\horizon}$, this representation recovers the complete observation prediction $\estimated{\observationProcess}_{\horizon|\horizon-1}$.

The main result of this paper is the derivation of an explicit, closed-form formula for the control sequence $\control$ over the horizon $\horizon$. This result is obtained by formulating and solving a dual optimal control problem.

\newP{Relationship to Literature} The representation~\eqref{eq:predictor} is justified based on the theory for Gaussian processes; cf.,~\cite[Corollary 1.10]{le2016brownian}.  Prior work on linear Gaussian models with non-Markovian structure generally follows three established estimation methodologies. 
These include the growing-state Kalman filter, which utilizes a growing state vector to track the system's full trajectory history~\cite[Chapter 8.2]{bar2001estimation},~\cite{tang2024augmented}; 
batch smoothing~\cite{bode1950smoothing}, which leverages the joint Gaussian structure of the states and observations to compute conditional means via closed-form matrix expressions;
and the causal Wiener-Hopf filter~\cite{wiener1949extrapolation}, which employs block factorizations to derive the optimal linear estimator while strictly enforcing causality~\cite{kailath2000linear}.
A more detailed overview of these classical approaches is provided in \Sec{sec:classical-methods}. An alternate class of non-Markovian processes include the reciprocal processes~\cite{white2017stochastic,moura2002recursive,levy2002modeling}.  These are not explicitly addressed because reciprocal processes admit non-causal dependencies, which are outside the causal framework considered here.

\section{Duality Theory}
\label{sec:duality}

\subsection{Dual Control System}
\label{sec:dualcontrolsys}

To develop the duality theory, we introduce a dual control system that 
runs backward in time, in contrast to the forward structure of the 
state process $X$. For this purpose, consider two deterministic processes as follows:
\begin{align*}
  \text{(dual state process)} &\quad \dualstate\defined\set{\dualstate_\initial, \dualstate_1,\hdots,\dualstate_{\horizon-1},\dualstate_{\horizon}},\\
  \text{(control process)}& \quad \control\defined\{\control_\initial,\control_1,\hdots,\control_{\horizon-1}\},
\end{align*}
where $\dualstate$ is $\reals^\stateDimension$-valued and $\control$ is $\reals^\observationDimension$-valued.  The lower case notation is used to stress the fact that these are deterministic processes. 
The space of deterministic control inputs is denoted as $\controlSpace 
\defined \{u = \{u_0, \ldots, u_{\horizon-1}\} : u_t \in \reals^\observationDimension,  
\; t \in \timeSpace^-\}$.
While the state process $X$ has a causal forward structure, the dual state process $\dualstate$ solves a backward difference equation ($\BDE$) as follows:
\begin{equation}
  \dualstate_{\timeSpace_-} = \transitionModel^\transpose \dualstate_{\timeSpace} + \observationModel^\transpose\control_{\timeSpace_- },\quad \dualstate_{\horizon}=f,
  \label{eq:dual_BDE}
\end{equation}
where $f$ is the terminal condition of the dual state process $\dualstate$ at time $\horizon$.
Recall $\observationModel\in\reals^{\horizon\observationDimension\times\horizon\stateDimension}$ is a block-diagonal matrix, and $\transitionModel\in\reals^{\horizon\stateDimension\times\horizon\stateDimension}$ is a block lower-triangular matrix.  
Therefore, $\transitionModel^\transpose$ is a block upper-triangular matrix and~\eqref{eq:dual_BDE} is well-posed (i.e., $y_t$ is uniquely determined by $f$ and the control sequence).  
Equation~\eqref{eq:dual_BDE} is referred to as the dual control system.

\subsection{Dual Optimal Control Problem}

For a fixed control $\control \in \controlSpace$ and a terminal condition $\dualstate_\horizon = f$, let $\dualstate$ denote the corresponding solution of~\eqref{eq:dual_BDE}. 
We define an optimal control type cost associated with this trajectory $\dualstate$ as follows:
\begin{equation}
  \dualCost_\horizon(\control;f) = \frac{1}{2}\abs{\dualstate_\initial}_{\initialCovariance}^2 + \sum_{t=1}^{\horizon}\frac{1}{2} \abs{\dualstate_t}_{\covarianceProcessNoise_t}^2 + \sum_{t=0}^{\horizon-1} \frac{1}{2} \abs{\control_t}_{\covarianceObservationNoise_t}^2,
  \label{eq:dual-cost}
\end{equation}
where $\abs{v}_{\dummyMatrix}^2 \defined v^\transpose \dummyMatrix v$ for any vector $v$ and any positive (semi)definite matrix $\dummyMatrix$.
The relationship to the estimation objective \eqref{eq:predictor} is given in the following theorem.

\begin{theorem}[Duality principle]\label{thm:duality-principle}
 Consider an estimator
  \begin{equation}
    \estimator_\horizon \defined \dualstate_\initial^\transpose\initialMean - \sum_{t=0}^{\horizon-1} \control_t^\transpose \observationProcess_{t},
    \label{eq:estimator}
  \end{equation}
  where $\dualstate$ is the solution of the dual control system~\eqref{eq:dual_BDE} 
for a control input $\control\in\controlSpace$ with terminal condition 
$\dualstate_\horizon = f\in\reals^\stateDimension$.
  Then 
  \begin{equation*}
    \dualCost_\horizon(\control;f) = \Expectation{\frac{1}{2}\abs{f^\transpose \stateProcess_{\horizon} - \estimator_\horizon}^2}.
  \end{equation*}
\end{theorem}
\begin{proof}
See Appendix~\ref{appdx:duality-principle}.
\end{proof}

\begin{remark}[Interpretation of the duality principle]
  The duality principle establishes that the optimal control cost $\dualCost_\horizon(\control;f)$ is equal to the mean-squared error of the estimator $\estimator_\horizon$ for predicting $f^\transpose \stateProcess_{\horizon}$.
Therefore, minimizing the control cost over $\control$ is equivalent to 
finding the best affine predictor of $f^\transpose \stateProcess_{\horizon}$ 
based on the observations $\observationProcess_{\initial:\horizon-1}$.
\end{remark}

The duality principle provides a control-theoretic approach to
compute the conditional expectation by solving an optimal control
problem given as follows:

\newP{Dual Optimal Control Problem}
\begin{equation}\label{eq:dual-optimal-control}
  \min_{\control \in \controlSpace} \; \dualCost_\horizon(\control;f)
  \quad \text{subject to} \quad\text{dual control system }\eqref{eq:dual_BDE}
\end{equation}

\section{Main Results}
\label{sec:results}

\subsection{Formula for optimal control}
\label{sec:dual-filter}

To solve the dual optimal control problem \eqref{eq:dual-optimal-control}, we introduce the momentum process $\momentum$, which is a deterministic $\reals^\stateDimension$-valued costate process associated with $\dualstate$:
\begin{equation*}
  \text{(momentum)} \quad \momentum\defined\set{\momentum_\initial, \momentum_1,\hdots,\momentum_{\horizon-1}, \momentum_{\horizon}}.
\end{equation*}
The optimality conditions for the control sequence are given in the following theorem:
\begin{theorem}[Optimal Control]\label{thm:hamilton}
  Consider the optimal control problem~\eqref{eq:dual-optimal-control}.
  \begin{subequations}
   Then the optimal control is given by the following forward-backward system:
    \begin{align}
      &\text{(backward)} & \dualstate_{\timeSpace_-} &= \transitionModel^\transpose \dualstate_{\timeSpace} + \observationModel^\transpose \control_{\timeSpace_-},\quad\dualstate_{\horizon}=f\label{eq:backward}\\
      &\text{(forward)} & \momentum_{\timeSpace} &= \transitionModel \momentum_{\timeSpace_-} + \covarianceProcessNoise \dualstate_{\timeSpace},\quad \momentum_\initial = \initialCovariance \dualstate_\initial\label{eq:forward}\\
      &\text{(opt. control)} & \control_{t} &= - \covarianceObservationNoise_t^\inverse\observationModel_t \momentum_{t},\quad \forall~t \in \timeSpace_- \label{eq:optimal_control} 
    \end{align}
    \label{eq:hamilton}
  \end{subequations}
  \vspace{-1.5\baselineskip}
\end{theorem}
\begin{proof}
See Appendix~\ref{appdx:hamilton}.
\end{proof}

\begin{remark}[Inversion of \(\horizon\)-dimensional Matrix] 
  The momentum process $\momentum$ serves as the costate variable in the 
  dual optimal control problem. 
  The optimal control formula~\eqref{eq:optimal_control} expresses $\control_t$ directly in terms of $\momentum_t$, without requiring the inversion of $\horizon$-dimensional matrix---the key structural property that yields the linear complexity of the dual filter.
\end{remark}


\subsection{Dual Filter Algorithm}

The dual filter (Algorithm~\ref{alg:dual-filter}) is an adjoint-based 
iterative procedure derived from the optimality conditions in~\Thm{thm:hamilton}.  
Each iteration consists of a 
backward pass to update the dual state $\dualstate$, a forward pass to 
update the momentum $\momentum$, and a gradient-based update of the 
control sequence $\control$. 
The optimal next-step prediction $\estimated{\observationProcess}_{\horizon|\horizon-1}$ in \eqref{eq:estimation-problem} is recovered by 
running the algorithm with $f$ set to each row of the observation 
matrix $\observationModel_\horizon$.

\newP{Complexity Analysis}
The backward and forward passes each require $\Order{\horizon\order\stateDimension}$ operations;
together with the $\Order{\horizon\observationDimension}$ control update, each iteration costs $\Order{\horizon(\order\stateDimension+\observationDimension)}$, linear in the context length $\horizon$ for a fixed order of $\order$ and fixed dimensions $\stateDimension$ and $\observationDimension$.
The memory cost is $\Order{\horizon\stateDimension}$ for storing the dual state and momentum trajectories, and $\Order{\horizon\observationDimension}$ for storing the observation trajectory, which is also linear in $\horizon$ for fixed dimensions.  In numerical simulations, typically a small number of iterations suffice for convergence as described in \Sec{sec:numerics}..

\begin{algorithm}[t] 
  \caption{Dual Filter}
  \label{alg:dual-filter}
  \begin{description}
    \item \textbf{Requirement}: Model parameters $\transitionModel$, $\observationModel$, $\initialMean$, $\initialCovariance$, $\covarianceProcessNoise$, $\covarianceObservationNoise$ 
    \item \textbf{Input}: terminal condition $f$; and observations $\observationProcess_{\initial:\horizon-1}$
    \item \textbf{Output}: optimal estimator $\estimator_\horizon$
  \end{description}
  \begin{enumerate}
    \item \textbf{Initialization}: Set $\control^{(0)}_t = 0$ for all $t\in\set{0,1,\hdots,\horizon-1}$

    \item \textbf{Iterate} for $\iteration = 0,1,2,\ldots$ until convergence:
    \begin{itemize}
      \item \emph{Backward pass} (dual state update): Set $\dualstate^{(\iteration)}_{\horizon}=f$;
      
        for $t=\horizon-1,\ldots,\initial$:
        \begin{equation*}
          \dualstate^{(\iteration)}_t = \sum_{s=1}^{\min(\order, \horizon-t)}\transitionModel_{t+s,s}^\transpose\dualstate^{(\iteration)}_{t+s}
            + \observationModel_t^\transpose\control^{(\iteration)}_t
        \end{equation*}

      \item \emph{Forward pass} (momentum update): Set
        $\momentum^{(\iteration)}_{\initial}=\initialCovariance\dualstate^{(\iteration)}_{\initial}$;

        for $t=1,\ldots,\horizon$:
        \begin{equation*}
          \momentum^{(\iteration)}_{t} = \sum_{s=1}^{\min(\order,t)}\transitionModel_{t,s}\momentum^{(\iteration)}_{t-s} + \covarianceProcessNoise_t\dualstate^{(\iteration)}_{t}
        \end{equation*}

      \item \emph{Control update} (gradient descent):
        \begin{equation*}
          \control^{(\iteration+1)}_t = \control^{(\iteration)}_t
            - \stepsize^{(\iteration)}_t\left(\control^{(\iteration)}_t + \covarianceObservationNoise_t^\inverse\observationModel_t\momentum^{(\iteration)}_t\right)\quad\forall~t\in\timeSpace_-
        \end{equation*}
        where $\stepsize^{(\iteration)}_t\succ 0$ is determined by the L-BFGS \cite{byrd2016stochastic} with the line search satisfying the strong Wolfe condition.
    \end{itemize}
    \item \textbf{Output}: The optimal estimator is given by
    \begin{equation*}
      \estimator_\horizon = \dualstate_{\initial}^\transpose\initialMean
      - \sum_{t=0}^{\horizon-1} \control_t^\transpose\observationProcess_t
    \end{equation*}
    where $\dualstate_{\initial}=\dualstate_{\initial}^{(\iteration)}$ and $\control=\control^{(\iteration)}$ are the converged solution of the initial dual state and control, respectively. 
  \end{enumerate}
\end{algorithm}

\subsection{Relationship to Transformer-like Architectures}
\label{sec:transformer-arch-relation}

The proposed dual filter bears a structural resemblance to the layer-wise 
operations of a decoder-only transformer, which we now describe. Two 
structural parallels are noted. First, both the transformer and the dual 
filter operate on a $\stateDimension \times \horizon$ data structure, 
preserving the dimensionality and length of the input sequence. Second, 
the iterative dual filter, comprising backward and forward passes to 
update the dual state $\dualstate$ and momentum $\momentum$, mirrors 
the layer-wise transformations of a transformer, with the momentum 
sequence $\momentum$ playing the role of the embedding sequence $\rho$. 
Within this framework, the control sequence $\control$ serves as the 
linear Gaussian analogue of attention weights, weighting historical 
observations to minimize the mean-squared prediction error. Unlike 
softmax attention weights, the optimal control $\control^*$ is obtained 
via the explicit formula~\eqref{eq:optimal_control}. The correspondence 
is depicted in~\Fig{fig:transformer-objective}, and causality is 
inherently enforced since $\estimator_\horizon$ depends only on past 
observations $\observationProcess_{\initial:\horizon-1}$.

\begin{figure}[t]
  \centering
  \begin{tikzpicture}
    \input{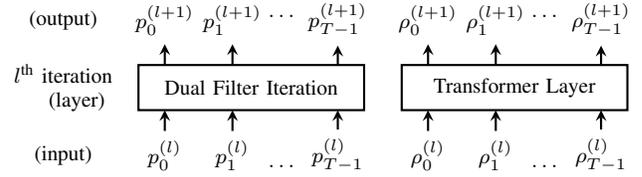}
  \end{tikzpicture}
  \caption{
    Structural correspondence between the dual filter iteration and a self-attention layer in a decoder-only transformer. 
    (left) The proposed iteration propagates the momentum sequence from $\momentum^{(\iteration)}_{\timeSpace_-}$ to $\momentum^{(\iteration+1)}_{\timeSpace_-}$. 
    (right) A transformer layer maps input embeddings $\rho^{(\iteration)}_{\timeSpace_-}$ to updated ones $\rho^{(\iteration+1)}_{\timeSpace_-}$. 
  }
  \label{fig:transformer-objective}
  \vspace{\baselineskip}
\end{figure}

\section{Comparison with Classical Methods}
\label{sec:classical-methods}

\subsection{Growing-state Kalman Filter}
Following~\cite{bar2001estimation} and~\cite{tang2024augmented}, the growing state estimate is defined as $\estimated{\stateProcess}_{\initial:t|t'} \defined \Expectation{\stateProcess_{\initial:t}\given\observationProcess_{\initial:t'}}$, and the error covariance be defined as $\covarianceError_{\initial:t|t'} \defined \Covariance{\stateProcess_{\initial:t}-\estimated{\stateProcess}_{\initial:t|t'}\given\observationProcess_{\initial:t'}}$ for $0\leq t'\leq t\leq\horizon$.

At each time step $t$, the Kalman filter is implemented through the following two-step recursive procedure:
\begin{enumerate}
  \item \emph{Correction Step:}
    \begin{itemize}
    \item Compute Kalman gain: 
    \begin{equation*}
      \KalmanGain_t = \covarianceError_{\initial:t|t-1}\adjusted{\observationModel}_t^\transpose(\adjusted{\observationModel}_t\covarianceError_{\initial:t|t-1}\adjusted{\observationModel}_t^\transpose+\covarianceObservationNoise_t)^\inverse
    \end{equation*}
    where $\adjusted{\observationModel}_t\defined \begin{bmatrix} 0 & \observationModel_t \end{bmatrix}_{\observationDimension\times (t+1)\stateDimension}$ is the adjusted observation matrix to ensure dimensional consistency.
    \item Update state estimate:
    \begin{equation*}
      \estimated{\stateProcess}_{\initial:t|t} = \estimated{\stateProcess}_{\initial:t|t-1} + \KalmanGain_t (\observationProcess_t - \adjusted{\observationModel}_t \estimated{\stateProcess}_{\initial:t|t-1})
    \end{equation*}
    where $\observationProcess_t$ is the newly obtained observation at time $t$.
    \item Update error covariance:
    \begin{equation*}
      \covarianceError_{\initial:t|t} = (\identity{(t+1)\stateDimension} - \KalmanGain_t \adjusted{\observationModel}_t) \covarianceError_{\initial:t|t-1}
    \end{equation*}
    where $\identity{\mathsf{n}}$ is an identity matrix with dimensions $\mathsf{n}\times\mathsf{n}$.
  \end{itemize}

  \item \emph{Transition Step:}
  \begin{itemize}
    \item Predict state estimate:
    \begin{equation*}
      \estimated{\stateProcess}_{\initial:t+1|t} = \begin{bmatrix}
        \hfill\estimated{\stateProcess}_{\initial:t|t} \\
        \adjusted{\transitionModel}_{t+1} \estimated{\stateProcess}_{\initial:t|t}
      \end{bmatrix}
    \end{equation*}
    where the adjusted transition matrix $\adjusted{\transitionModel}_{t+1} \in \reals^{\stateDimension \times (t+1)\stateDimension}$ is defined as $ \adjusted{\transitionModel}_{t+1} \defined \begin{bmatrix} 0 & \transitionModel_{t+1,\order} & \cdots & \transitionModel_{t+1,1} \end{bmatrix}$ which is appropriately padded with zeros or truncated to ensure dimensional consistency.
    The size of the growing state estimation vector $\estimated{\stateProcess}_{\initial:t+1|t}$ is $(t+1)\stateDimension$.
    \item Predict error covariance:
    \begin{equation*}
      \covarianceError_{\initial:t+1|t} = \begin{bmatrix}
        \hfill\covarianceError_{\initial:t|t} & \phantom{\adjusted{\transitionModel}_{t+1}}\covarianceError_{\initial:t|t}\adjusted{\transitionModel}_{t+1}^\transpose\hfill\\
        \adjusted{\transitionModel}_{t+1}\covarianceError_{\initial:t|t} &
        \adjusted{\transitionModel}_{t+1} \covarianceError_{\initial:t|t} \adjusted{\transitionModel}_{t+1}^\transpose + \covarianceProcessNoise_{t+1}
      \end{bmatrix}
    \end{equation*}
    where the shape of the growing error covariance matrix is $(t+1)\stateDimension \times (t+1)\stateDimension$. 
  \end{itemize}
\end{enumerate}
The initial state and error covariance estimates are given by
\begin{equation*}
  \estimated{\stateProcess}_{\initial:\initial|-1} = \estimated{\stateProcess}_{\initial|-1} \defined \initialMean,\quad \covarianceError_{\initial:\initial|-1} = \covarianceError_{\initial|-1} \defined \initialCovariance
\end{equation*}
and the prediction simply follows 
\begin{equation}
  \estimated{\observationProcess}_{\horizon|\horizon-1} = \observationModel_{\horizon} \estimated{\stateProcess}_{\horizon|\horizon-1}
  \label{eq:kalman-prediction}
\end{equation}
where $\estimated{\stateProcess}_{\horizon|\horizon-1}$ is computed in transition step at $t=\horizon-1$.

\newP{Complexity Analysis}
At time $t$, the filter maintains a state vector of size $(t+1)\stateDimension$ and a covariance matrix of size $(t+1)\stateDimension \times (t+1)\stateDimension$, making each step cost $\Order{t^2 \stateDimension^2}$. Summing over $t \in \timeSpace_-$ yields a total time complexity of
\begin{equation*}
\sum_{t \in \timeSpace_-} \Order{t^2 \stateDimension^2}
= \Order{\horizon^3 \stateDimension^2}.
\end{equation*}
The memory complexity is $\Order{\horizon^2 \stateDimension^2}$, due to storing a covariance matrix of size $(\horizon+1)\stateDimension \times (\horizon+1)\stateDimension$.



\subsection{Batch Smoothing}
\label{paragraph:batch-smoothing}

Since $(\stateProcess, \observationProcess)$ is jointly Gaussian, the smoothing distribution $\Probability{\stateProcess_{\timeSpace_-} \given \observationProcess_{\timeSpace_-}}$ is also Gaussian~\cite{anderson2005optimal}.
Denote the smoothed state estimates as
\begin{equation*}
  \estimated{\stateProcess}_{t|\horizon-1} \defined \Expectation{\stateProcess_t \given \observationProcess_{\timeSpace_-}},\quad \forall t \in \timeSpace,
\end{equation*}
and the means $\mean{\stateProcess}_{\timeSpace_-} \defined \Expectation{\stateProcess_{\timeSpace_-}}$, $\mean{\observationProcess}_{\timeSpace_-} \defined \Expectation{\observationProcess_{\timeSpace_-}}$ together with the covariances $\covState \defined \Covariance{\stateProcess_{\timeSpace_-}},~\covObs \defined \Covariance{\observationProcess_{\timeSpace_-}},~\covStateObs \defined \Covariance{\stateProcess_{\timeSpace_-}, \observationProcess_{\timeSpace_-}}$.
The optimal smoothing prediction is given by the conditional mean, which has an affine form in the observations $\observationProcess_{\timeSpace_-}$:
\begin{equation}
  \estimated{\stateProcess}_{\timeSpace|\horizon-1}
  = \Expectation{\transitionModel\stateProcess_{\timeSpace_-} \given \observationProcess_{\timeSpace_-}}
  = \smoothingWeight \observationProcess_{\timeSpace_-} + \smoothingBias,
  \label{eq:affine-estimator}
\end{equation}
where $\smoothingWeight \in \reals^{\horizon\stateDimension \times \horizon\observationDimension}$ is the smoothing projection weight and $\smoothingBias \in \reals^{\horizon\stateDimension}$ is the mean correction.
By imposing the unbiasedness condition, i.e., $\Expectation{\estimated{\stateProcess}_{\timeSpace_-|\horizon-1}} = \Expectation{\stateProcess_{\timeSpace_-}}$, one obtains the bias term:
\begin{equation}
  \smoothingBias = \transitionModel\mean{\stateProcess}_{\timeSpace_-} - \smoothingWeight\mean{\observationProcess}_{\timeSpace_-}.
  \label{eq:unbiasedness}
\end{equation}
Next, by invoking the orthogonality condition, namely that the estimation error is uncorrelated with the centered observations, i.e., $\Expectation{(\stateProcess_{\timeSpace_-} - \estimated{\stateProcess}_{\timeSpace_-|\horizon-1})(\observationProcess_{\timeSpace_-}-\mean{\observationProcess}_{\timeSpace_-})^\transpose} = 0$, one obtains the normal equations for $\smoothingWeight$:
\begin{align}
  \transitionModel\covStateObs = \smoothingWeight\covObs \implies \smoothingWeight = \transitionModel\covStateObs (\covObs)^\inverse.\label{eq:orthogonality}
\end{align}
The details of the computation of $\mean{\stateProcess}_{\timeSpace_-}$, $\mean{\observationProcess}_{\timeSpace_-}$, $\covStateObs$, and $\covObs$ are given in Appendix~\ref{appdx:mean-covariance}.
The prediction using the batch smoothing approach is then
\begin{subequations}
  \label{eq:smoothing-prediction}
  \begin{align}
    \estimated{\stateProcess}_{\timeSpace|\horizon-1}
    &= \transitionModel\!\left(\mean{\stateProcess}_{\timeSpace_-}
       + \covStateObs(\covObs)^\inverse
         (\observationProcess_{\timeSpace_-} - \mean{\observationProcess}_{\timeSpace_-})\right),
    \label{eq:smoothing-state-prediction}\\
    \estimated{\observationProcess}_{\horizon|\horizon-1}
    &= \adjusted{\observationModel}_{\horizon}\,\estimated{\stateProcess}_{\timeSpace|\horizon-1},
  \end{align}
  where $\adjusted{\observationModel}_{\horizon} \defined \begin{bmatrix} 0 & \observationModel_{\horizon} \end{bmatrix}_{\observationDimension\times\horizon\stateDimension}$ is the adjusted observation matrix for time index $\horizon$, and the initial prediction is given by $\estimated{\observationProcess}_{\initial|-1} = \observationModel_{\initial}\initialMean$ for completeness.
\end{subequations}
One may recover the control trajectory corresponding to each entry of $\estimated{\observationProcess}_{\horizon|\horizon-1}$ by taking $f$ as each row of $\observationModel_\horizon$ based on the duality principle (\Thm{thm:duality-principle}). The control is given by
\begin{equation}
  \control_t^\transpose = -f^\transpose(\smoothingWeight)_{\horizon,t},
  \quad \forall~t\in\timeSpace_-,
  \label{eq:classical-control}
\end{equation}
where $(\smoothingWeight)_{\horizon,t}\in\reals^{\stateDimension\times\observationDimension}$ is the $(\horizon,t+1)$-th block of $\smoothingWeight$, corresponding to the projection from $\observationProcess_t$ to $\estimated{\stateProcess}_{\horizon|\horizon-1}$.

\newP{Complexity Analysis}
The computational bottleneck is the direct inversion of the joint observation covariance matrix $\covObs \in \reals^{\horizon\observationDimension \times \horizon\observationDimension}$. 
Although derived from structured models, $\covObs$ is generally dense under non-Markovian dynamics, resulting in a time complexity of $\Order{\horizon^3 \observationDimension^3}$. 
Memory requirements scale as $\Order{\horizon^2 (\observationDimension + \stateDimension)^2}$ to accommodate the storage of the state covariance $\covState$, the observation covariance $\covObs$, and the cross-covariance $\covStateObs$.


\subsection{Causal Wiener--Hopf Filter}
\label{paragraph:wiener-hopf}

The decoder-only transform has a causal structure, where the prediction at time $t$ can only depend on observations up to time $t-1$.
For this reason, it is also useful to consider the causal counterpart of the batch smoothing approach, which is the causal Wiener--Hopf filter~\cite{wiener1949extrapolation}.
Denote the filtered state estimates stacked as
\begin{equation*}
  \estimated{\stateProcess}_{\timeSpace|\timeSpace_-}
  \defined
  \begin{bmatrix} \estimated{\stateProcess}_{1|0} \\ \vdots \\ \estimated{\stateProcess}_{\horizon|\horizon-1} \end{bmatrix},~
  \estimated{\stateProcess}_{t|t-1} \defined \Expectation{\stateProcess_t \given \observationProcess_{\initial:t-1}},~ \forall t \in \timeSpace.
\end{equation*}
The conditional mean estimated from the filter can be represented in an affine form as follows:
\begin{align}
  \estimated{\stateProcess}_{\timeSpace|\horizon-1}
  = \Expectation{\transitionModel\stateProcess_{\timeSpace_-} \given \observationProcess_{\timeSpace_-}}
  = \filteringWeight \observationProcess_{\timeSpace_-} + \filteringBias,
  \label{eq:affine-estimator-filter}
\end{align}
where $\filteringWeight \in \reals^{\horizon\stateDimension \times \horizon\observationDimension}$ is the filtering projection weight and the bias term \(\filteringBias\) is given by
\begin{equation}
  \filteringBias = \transitionModel \mean{\stateProcess}_{\timeSpace_-} - \filteringWeight \mean{\observationProcess}_{\timeSpace_-}.
  \label{eq:unbiasedness-filter}
\end{equation}
The causality constraint is enforced by requiring the projection $\filteringWeight$ to be block lower-triangular ($(\filteringWeight)_{t,t'} = 0$ for $t' > t$), so that $\estimated{\stateProcess}_{t|t-1}$ depends only on $\observationProcess_{\initial:t-1}$ for all $t\in\timeSpace$.
To enforce causality efficiently, we use block Cholesky factorization~\cite{chen2013block} to factor $\covObs = \innovations\innovations^\transpose$ with block lower-triangular $\innovations \in \reals^{\horizon\observationDimension\times\horizon\observationDimension}$, and combine with the orthogonality condition~\eqref{eq:orthogonality} to obtain
\begin{equation}
  \filteringWeight
  = \Bigl[\transitionModel\covStateObs (\innovations^\transpose)^\inverse\Bigr]_{\mathrm{lower}}
    \innovations^{\inverse},
  \label{eq:wiener-hopf-projection}
\end{equation}
where $[\,\cdot\,]_{\mathrm{lower}}$ denotes block lower-triangular projection. The computation of the covariances $\covStateObs$ and $\covObs$ follows Appendix~\ref{appdx:mean-covariance}.
Combining the affine form~\eqref{eq:affine-estimator-filter} with the bias formula~\eqref{eq:unbiasedness-filter} and the projection~\eqref{eq:wiener-hopf-projection}, the prediction is
\begin{subequations}
  \label{eq:filtering-prediction}
  \begin{align}
    \estimated{\stateProcess}_{\timeSpace|\timeSpace_-}
    &= \transitionModel\mean{\stateProcess}_{\timeSpace_-}
       + \filteringWeight(\observationProcess_{\timeSpace_-} - \mean{\observationProcess}_{\timeSpace_-}),
    \label{eq:filtering-state-prediction}\\
    \estimated{\observationProcess}_{\horizon|\horizon-1}
    &= \adjusted{\observationModel}_{\horizon}\,\estimated{\stateProcess}_{\timeSpace|\timeSpace_-},
  \end{align}
\end{subequations}
where $\filteringWeight$ is given by~\eqref{eq:wiener-hopf-projection}.
To recover the control trajectory, one takes $f$ as each row of $\observationModel_\horizon$ based on the duality principle (\Thm{thm:duality-principle}), yielding~\eqref{eq:classical-control} with $\smoothingWeight$ replaced by $\filteringWeight$.

\newP{Complexity Analysis}
Block Cholesky factorization of $\covObs$ dominates the computational cost, maintaining the $\Order{\horizon^3 \observationDimension^3}$ time and $\Order{\horizon^2 (\observationDimension + \stateDimension)^2}$ memory complexities seen in batch smoothing.

\begin{remark}
  While the growing-state Kalman filter, batch smoothing, the causal Wiener-Hopf filter, and the proposed dual filter utilize distinct computational procedures and exhibit different complexities, they are fundamentally equivalent in terms of optimality.
  Because each approach solves the same underlying estimation problem, all four methods yield identical results for the prediction $\estimated{\observationProcess}_{\horizon|\horizon-1}$.
\end{remark}


\begin{figure*}[t]
  \medskip
  \centering
  \includegraphics[width=0.95\linewidth]{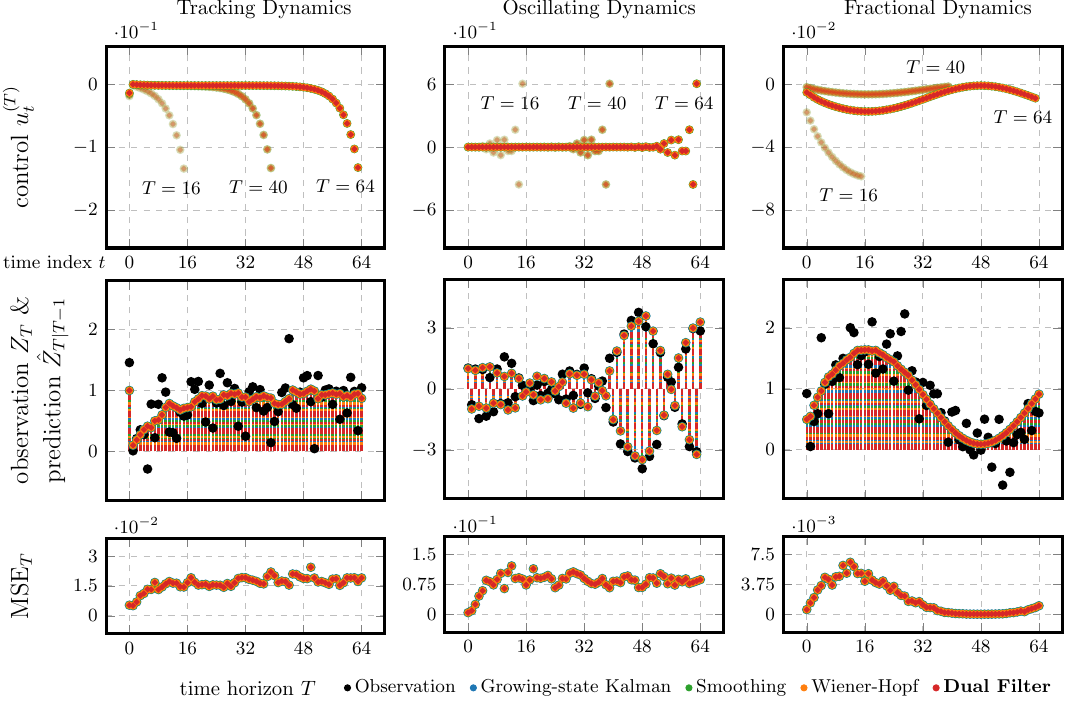}
  \caption{
    Numerical comparison of the proposed dual filter with classical estimation methods across three dynamical systems. 
    The columns correspond to the tracking, oscillating, and cumulative fractional dynamics, respectively. 
    The first row shows the controls $\control_t^{(\horizon)}$ for $\horizon = 16$, $40$, and $64$, obtained from batch smoothing, the causal Wiener-Hopf filter, and the dual filter, which overlap exactly.
    The second row presents one realization of the predictions $\estimated{\observationProcess}_{\horizon|\horizon-1}$ alongside the noisy observations $\observationProcess_\horizon$, while the third row shows the mean squared error $\error_\horizon$ over a batch containing \(100\) trajectories; 
    in both cases, the results are indistinguishable across methods.
    Results are color-coded as follows: observations (black), growing-state Kalman filter (blue), batch smoothing (green), causal Wiener-Hopf filter (orange), and dual filter (red).
  }
  \label{fig:numerics}
\end{figure*}

\section{Numerical Examples}
\label{sec:numerics}

In this section, we evaluate the performance of the proposed dual filter and compare it against classical estimation methods through a series of numerical experiments. 
For clarity of presentation, we consider three representative classes of dynamics with dimensions given by $\stateDimension = \observationDimension = 1$.

\smallskip

\begin{enumerate}
  \item Tracking Dynamics:
  \begin{equation*}
    \stateProcess_{t} =
    \begin{cases}
      (1 - \trackingRate)\stateProcess_{t-1} + \trackingRate \stateProcess_0 + \processNoise_{t}, & 1 < t \leq \horizon \\
      \hfill \trackingRate \stateProcess_0 + \processNoise_{t}, & t = 1
    \end{cases}
  \end{equation*}
  where $\trackingRate\in(0,1)$ is the tracking rate that controls how quickly the system tracks the initial state that depends on the entire history of the state process, and thus corresponds to a full-order system with $\order = \horizon$.
  The observation model is set to $\observationModel_t \equiv 1 \ \forall~t\in\set{0,1,\ldots,\horizon}$.

  \smallskip

  \item Oscillating Dynamics:
  \begin{equation*}
    \stateProcess_{t} =
    \begin{cases}
      (- 2 \cos \oscillationAngle)\stateProcess_{t-1} - \stateProcess_{t - 2} + \processNoise_{t}, & 1 < t \leq \horizon \\
      \phantom{2}(- \cos \oscillationAngle)\stateProcess_{t-1} \hfill + \processNoise_{t}, & t = 1
    \end{cases}
  \end{equation*}
  where $\oscillationAngle>0$ represents the rotating angle that determines the oscillating frequency of the system's response. 
  This system is specifically selected to demonstrate the model's capabilities in a discrete-time setting. 
  Notably, the negative sign associated with the cosine term induces a sign-flipping behavior at each time step, a feature intrinsic to the discrete-time formulation.
  Similar to the tracking system, the observation model is set to $\observationModel_t \equiv 1 \ \forall~t\in\set{0,\ldots,\horizon}$, and the system corresponds to a fixed-order case with $\order = 2$.

  \smallskip

  \item Cumulative Fractional Dynamics:
  \begin{equation*}
    \stateProcess_{t} = \sum_{s=1}^{t} \frac{1}{\fractionCoefficient_{t,s}} \stateProcess_{t-s} + \processNoise_{t}, \quad 1 \leq t \leq \horizon
  \end{equation*}
  where the fractional coefficients $\fractionCoefficient_{t,s}=(t-s+1)^\polynomialOrder$ which decay polynomially with the order $\polynomialOrder>1$, and thus corresponds to a full-order system with $\order = \horizon$.
  The observation model is set to $\observationModel_t = \frac{1}{2} (1 + 0.9\sin(\discreteFrequency t))$ with discrete frequency $\discreteFrequency>0$, which introduces a time-varying observation model.
\end{enumerate}

\smallskip


\noindent The systems are evaluated within a noisy regime characterized by specific stochastic parameters. 
The initial state is defined by a mean of $\initialMean = 1$ and a variance of $\initialCovariance = 5 \cdot 10^{-3}$. 
Furthermore, the system is subjected to time-invariant process and observation noise; the process noise variance is fixed at $\covarianceProcessNoise_t \equiv 5 \cdot 10^{-3}$, while the observation noise variance is set to $\covarianceObservationNoise_t \equiv 10^{-1}$. 
These statistical properties are assumed to remain constant throughout the horizon $t \in \set{0, 1, \ldots, \horizon}$.

\subsection{Estimation Accuracy \& Control Interpretation}

To assess estimation accuracy, we compare the predictions produced by the proposed dual filter with those obtained from classical methods, including the growing-state Kalman filter \eqref{eq:kalman-prediction}, batch smoothing \eqref{eq:smoothing-prediction}, and the causal Wiener-Hopf filter \eqref{eq:filtering-prediction}, over horizons ranging from $ \horizon = 0 $ to $ \horizon = 64 $ across the three systems described above. 
The system parameters are set to $\trackingRate = 0.1$, $\oscillationAngle = \frac{\pi}{18}$, $\polynomialOrder = 2$, and $\discreteFrequency = \frac{\pi}{32}$.
Across all systems, the predictions from the different methods are identical, indicating that the dual filter matches the performance of established approaches.
To provide a quantitative assessment, we evaluate the empirical mean squared error $\error_\horizon$ for each time horizon $\horizon$, defined as
\begin{equation*}
  \error_\horizon \defined \frac{1}{N}\sum_{n=1}^N\frac{1}{2}\abs{\observationModel_\horizon \stateProcess_\horizon^{(n)} - \estimated{\observationProcess}_{\horizon|\horizon-1}^{(n)}}^2
\end{equation*}
where the index $n$ denotes the $n$-th trajectory, and $N=100$ is the total number of trajectories.
Here, $ \estimated{\observationProcess}_{\horizon|\horizon-1}^{(n)}$ denotes the prediction from each method, and $ \stateProcess_\horizon^{(n)}$ denotes the underlying true state. 
The results, shown in the second and third rows of \Fig{fig:numerics}, demonstrate identical error profiles across all horizons, confirming that the proposed dual filter achieves estimation accuracy on par with classical approaches.

Beyond validating prediction accuracy, we analyze how the dual filter weights historical observations.
The control trajectories $\control_t^{(\horizon)}$ are examined across varying horizons $\horizon$. 
As illustrated in \Fig{fig:numerics}, the trajectories generated by the dual formulation align precisely with those from classical projection-based methods \eqref{eq:classical-control}, providing empirical support for the duality principle in \Thm{thm:duality-principle}.


The structure of controls reflects the temporal characteristics of each system. 
For the tracking dynamics ($\order = \horizon$), the control places significant weight on the most recent observations, rapidly decays over the intermediate past, and exhibits a distinct spike at the initial time step. 
This pattern indicates that the optimal predictor relies primarily on the immediate past while retaining a direct dependence on the initial state. 
In the oscillating dynamics ($\order = 2$), the control is localized to the near past and displays a sign-flipping pattern consistent with the system's oscillatory behavior. 
For the cumulative fractional dynamics, the control magnitude varies in accordance with the time-varying observation model $\observationModel_t$, assigning stronger weight to observations with higher signal strength.
These results show that the dual filter achieves optimal predictions while yielding interpretable, dynamics-adaptive control structures.

\subsection{Computational Complexity}

Finally, we evaluate computational efficiency by measuring FLOPs across varying time horizons using the \emph{PAPI} interface \cite{papi2025main}. 
As shown in \Fig{fig:complexity}, the proposed dual filter achieves over a $100\times$ speedup compared to classical full-order methods at $\horizon = 2^{14}$. 
This superior scalability makes the framework ideal for long-horizon applications, such as sequence modeling, where traditional cubic complexity is computationally prohibitive.

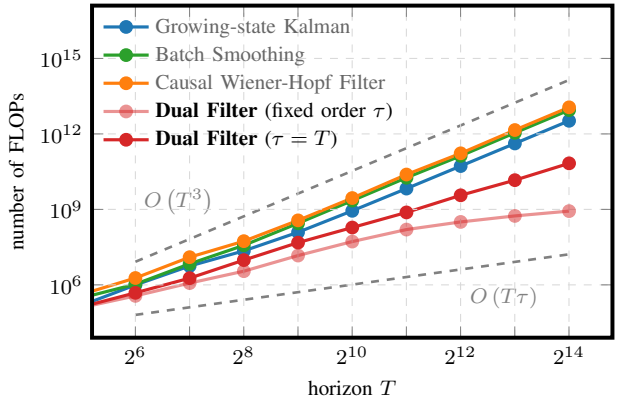
\begin{figure}
  \medskip
  \centering
  \footnotesize
  \begin{tabular}{lll}
    \toprule
    method & complexity & paradigm \\
    \midrule
    Growing-state Kalman Filter & $\Order{\horizon^3 \stateDimension^2}$ & recursive \\
    Batch Smoothing & $\Order{\horizon^3 \observationDimension^3}$ & batch \\
    Causal Wiener-Hopf Filter & $\Order{\horizon^3 \observationDimension^3}$ & batch \\
    Dual Filter & $\Order{\horizon (\order \stateDimension + \observationDimension)}$ & sequential$^*$ \\
    \bottomrule
  \end{tabular}

  \medskip

  \begin{tikzpicture}
    \pgfplotstableread[col sep=comma]{data/flops.csv}\complexityData

\def\minHorizonPower{6} 
\def\maxHorizonPower{14} 
\pgfmathsetmacro\diffHorizonPower{(\maxHorizonPower - \minHorizonPower)}
\pgfmathsetmacro\gapHorizonPower{\diffHorizonPower * 0.1}
\def\minHorizon{\fpeval{2^(\minHorizonPower - \gapHorizonPower)}}
\def\maxHorizon{\fpeval{2^(\maxHorizonPower + \gapHorizonPower)}}

\def\minFLOPsPower{5}
\def\maxFLOPsPower{16}
\pgfmathsetmacro\diffFLOPsPower{(\maxFLOPsPower - \minFLOPsPower)}
\pgfmathsetmacro\gapFLOPsPower{\diffFLOPsPower * 0.1}
\def\ymin{\fpeval{\minFLOPsPower - \gapFLOPsPower}}
\def\ymax{\fpeval{\maxFLOPsPower + \gapFLOPsPower}}

\pgfmathsetmacro{\flopsScaleOrder}{log2(10)}  

\typeout{gap Horizon Power: \gapHorizonPower}
\typeout{min Horizon: \minHorizon, max Horizon: \maxHorizon}

\begin{axis}[
  width = 8.5cm,
  height = 6cm,
  xlabel={horizon $\horizon$},
  ylabel={number of FLOPs},
  xmode=log,
  log basis x=2,
  xmin=\minHorizon, xmax=\maxHorizon,
  ymin=\ymin, ymax=\ymax,
  xtickten={6, 7, 8, 9, 10, 11, 12, 13, 14}, 
  xticklabels={$2^6$, , $2^8$, , $2^{10}$, , $2^{12}$, , $2^{14}$},
  ytick={6, 9, 12, 15}, 
  yticklabels={$10^6$, $10^9$, $10^{12}$, $10^{15}$},
  minor x tick num=1,           
  grid = both, 
  grid style={dashed, gray!30},
  legend style={
    at={(0,1)},                
    anchor=north west,         
    draw=none,                 
    fill=none,                 
    font=\footnotesize,        
  },
  legend cell align={left},
  line width=1.5pt,
  axis on top,
]

  \addplot[
    color=\kalmanColor,
    line width=\lineWidth,
    mark=*,
  ] table[
    x=horizon, 
    y expr=\thisrow{kalman}/\flopsScaleOrder,
  ] {\complexityData};
  \addlegendentry{\color{black!60}Growing-state Kalman}

  \addplot[
    color=\smoothingColor,
    line width=\lineWidth,
    mark=*,
  ] table[
    x=horizon, 
    y expr=\thisrow{smoothing}/\flopsScaleOrder,
  ] {\complexityData};
  \addlegendentry{\color{black!60}Batch Smoothing}

  \addplot[
    color=\causalColor,
    line width=\lineWidth,
    mark=*,
  ] table[
    x=horizon, 
    y expr=\thisrow{causal}/\flopsScaleOrder,
  ] {\complexityData};
  \addlegendentry{\color{black!60}Causal Wiener-Hopf Filter}
  
  \addplot[
    color=\dualColor,
    line width=\lineWidth,
    mark=*,
    opacity = 0.5,
  ] table[
    x=horizon, 
    y expr=\thisrow{dual-K}/\flopsScaleOrder,
  ] {\complexityData};
  \addlegendentry{\textbf{Dual Filter} (fixed order $\order$)}

  \addplot[
    color=\dualColor,
    line width=\lineWidth,
    mark=*,
  ] table[
    x=horizon, 
    y expr=\thisrow{dual-T}/\flopsScaleOrder,
  ] {\complexityData};
  \addlegendentry{\textbf{Dual Filter} ($\order=\horizon$)}

  \addplot [
    domain=64:16384, 
    dashed, 
    gray, 
    line width=1pt,
    samples=2
  ] {3*log2(x)/\flopsScaleOrder + 1.5}
  node [pos=0.1, above=8pt] {$\Order{\horizon^3}$};

  \addplot [
    domain=64:16384, 
    dashed, 
    gray, 
    line width=1pt,
    samples=2
  ] {log2(x)/\flopsScaleOrder + 3}
  node [pos=0.85, below=6pt] {$\Order{\horizon\order}$};

\end{axis}
  \end{tikzpicture}
  \caption{
    Computational complexity comparison. 
    FLOPs as a function of horizon $\horizon$ for the growing-state Kalman filter (blue), batch smoothing (green), causal Wiener-Hopf filter (orange), and the proposed Dual Filter (red). 
    Classical methods exhibit $\Order{\horizon^3}$ complexity, whereas the dual filter achieves linear complexity $\Order{\horizon}$ for fixed-order settings ($\order=2$, light red) and quadratic complexity $\Order{\horizon^2}$ for full-order settings ($\order=\horizon$, dark red). 
    Solid lines represent empirical measurements; dashed lines indicate theoretical trends. 
    $^*$Unlike the recursive Kalman filter, the dual filter is a sequential batch-processing algorithm that operates on the full observation sequence, a common dependency in batch methods.
  }
  \label{fig:complexity}
  \vspace{1.5\baselineskip}
\end{figure}

\section{Conclusions}
\label{sec:conclusion}

In this paper, we extended Kalman duality to non-Markovian linear Gaussian models, deriving a dual filtering algorithm that matches the estimation performance of classical methods, including batch smoothing and Wiener-Hopf filtering, while scaling linearly or quadratically with the time horizon, compared to the cubic scaling of classical approaches.
Numerical experiments confirm identical prediction accuracy and mean squared error across diverse dynamical systems. 
In concurrent work, we developed a dual filter for Hidden Markov Models~\cite{chang2025dual}. 
Future directions include extending this nonlinear framework to non-Markovian settings to capture long-range dependencies, investigating observation-dependent stochastic control policies~\cite{talebi2022duality} analogous to attention mechanisms, and exploring learning-based approaches in which transformers are trained to perform optimal filtering for unknown systems~\cite{du2023can}.


\section*{APPENDIX}
\renewcommand{\thesubsection}{\Alph{subsection}}
\numberwithin{equation}{subsection}
\setcounter{equation}{0} 

\subsection{Proof of Duality Principle (Theorem~\ref{thm:duality-principle})}
\label{appdx:duality-principle}

The duality principle is proved by first establishing the identity 
\begin{equation}\label{eq:pairing-identity}
  f^\transpose \stateProcess_{\horizon} - S_\horizon = \dualstate_\initial^\transpose(\stateProcess_\initial - \initialMean) + \sum_{t=1}^{\horizon}\dualstate_{t}^\transpose\processNoise_{t} + \sum_{t=0}^{\horizon-1}\control_t^\transpose\observationNoise_t
\end{equation}
and then taking the squared expectation of both sides, which yields the desired result since the three summands on the right-hand side are mutually uncorrelated zero-mean Gaussian random variables, so the expectation of their sum's square is the sum of their squared expectations.

\subparagraph{Step 1 (Pairing Identity)}
Using the observation process \eqref{eq:observation}, compact state process equation \eqref{eq:compact-process} and $\BDE$ \eqref{eq:dual_BDE}, we compute
\begin{align*}
  \dualstate_{\timeSpace}^\transpose\stateProcess_{\timeSpace}
  &= \dualstate_{\timeSpace_-}^\transpose\stateProcess_{\timeSpace_-} - \sum_{t\in\timeSpace_-}\control_t^\transpose(\observationProcess_t - \observationNoise_t) + \sum_{t\in\timeSpace}\dualstate_{t}^\transpose\processNoise_t
\end{align*}
which can be rearranged to
\begin{equation*}
  \dualstate_\horizon^\transpose\stateProcess_\horizon = \dualstate_\initial^\transpose\stateProcess_\initial - \sum_{t=0}^{\horizon-1}\control_t^\transpose\observationProcess_t + \sum_{t=0}^{\horizon-1}\control_t^\transpose\observationNoise_t + \sum_{t=1}^{\horizon}\dualstate_{t}^\transpose\processNoise_t
\end{equation*}
The above identity together with the terminal condition of dual state $\dualstate_\horizon=f$, and the definition of estimator $S_\horizon$ in \eqref{eq:estimator} gives \eqref{eq:pairing-identity}.

\subparagraph{Step 2 (Squared Expectation)}
Since $\stateProcess_\initial-\initialMean\sim\Gaussian(0,\initialCovariance)$, $\processNoise_t\sim\Gaussian(0,\covarianceProcessNoise_{t})$, $\observationNoise_t\sim\Gaussian(0,\covarianceObservationNoise_{t})$, and $\stateProcess_\initial,\processNoise,\observationNoise$ are mutually independent, all three summands on the right hand side of \eqref{eq:pairing-identity} are zero-mean and mutually uncorrelated.
Therefore,
\begin{equation*}
  \Expectation{\abs{f^\transpose\stateProcess_{\horizon}-S_\horizon}^2} = \abs{\dualstate_\initial}^2_{\initialCovariance} + \sum_{t=1}^{\horizon}\abs{\dualstate_{t}}^2_{\covarianceProcessNoise_t} + \sum_{t=0}^{\horizon-1}\abs{\control_t}^2_{\covarianceObservationNoise_t}
\end{equation*}
and this concludes the proof since it shows the quadratic cost defined in \eqref{eq:dual-cost} is exactly the mean squared error of the estimator $S_\horizon$ to the prediction $f^\transpose\stateProcess_\horizon$.
\hfill$\square$

\subsection{Proof of Optimal Control Formula (Theorem~\ref{thm:hamilton})}
\label{appdx:hamilton}

To establish the optimality condition, let $\control^\optimum$ denote the optimal control and consider an arbitrary perturbation $\Delta\control$. 
We define the perturbed control trajectory as $\control = \control^\optimum + \Delta\control$. 
By leveraging the linearity of the dual control system \eqref{eq:backward}, the resulting dual state $\dualstate$ admits the decomposition $\dualstate = \dualstate^\optimum + \Delta\dualstate$. 
Here, $\dualstate^\optimum$ and $\Delta\dualstate$ are the unique solutions to the dual dynamics corresponding to $\control^\optimum$ (with terminal condition $\dualstate^\optimum_\horizon = f$) and $\Delta\control$ (with terminal condition $\Delta\dualstate_\horizon = 0$), respectively.
The cost functional $\dualCost_\horizon(\control; f)$ can be expanded quadratically around $\control^\optimum$ as follows:
\begin{equation*}
  \dualCost_\horizon(\control; f) = \dualCost_\horizon(\control^\optimum; f) + \dualCost_\horizon(\Delta\control; 0) + \mathcal{I}_\horizon(\control^\optimum, \Delta\control)
\end{equation*}
where $\dualCost_\horizon(\control^\optimum; f)$ represents the minimum cost, $\dualCost_\horizon(\Delta\control; 0)$ is the second-order error term (always non-negative), and the first-order variation (cross terms) is given by:
\begin{align*}
  \mathcal{I}_\horizon(\control^\optimum, \Delta\control) &= (\dualstate_\initial^\optimum)^\transpose\initialCovariance(\Delta\dualstate_\initial)\\
  &\quad + \sum_{t=1}^{\horizon}(\dualstate_{t}^\optimum)^\transpose\covarianceProcessNoise_t(\Delta\dualstate_{t}) + \sum_{t=\initial}^{\horizon-1}(\control_t^\optimum)^\transpose\covarianceObservationNoise_t(\Delta\control_t)
\end{align*}
By invoking the definition of the forward momentum process $\momentum^\optimum$ from \eqref{eq:forward} and substituting the dual control system dynamics for $\Delta\dualstate$, we can simplify the first-order variation into the following inner product form:
\begin{equation*}
  \mathcal{I}_\horizon(\control^\optimum, \Delta\control) = \sum_{t=\initial}^{\horizon-1} \left( \observationModel_t \momentum^\optimum_{t} + \covarianceObservationNoise_t \control^\optimum_{t} \right)^\transpose (\Delta\control_t)
\end{equation*}
The optimality of $\control^\optimum$ requires that the first-order variation $\mathcal{I}_\horizon(\control^\optimum, \Delta\control)$ must vanish for any arbitrary perturbation $\Delta\control$. 
This stationarity condition implies $\observationModel_t \momentum^\optimum_{t} + \covarianceObservationNoise_t \control^\optimum_{t} = 0$ for all $t \in \set{\initial, \dots, \horizon-1}$, which directly yields the optimal control law \eqref{eq:optimal_control}. 
This completes the proof. \hfill$\square$

\subsection{Computation of Means and Covariances}
\label{appdx:mean-covariance}

Let the mean of states be denoted as
\begin{equation*}
  \mean{\stateProcess}_{t}\defined\Expectation{\stateProcess_{t}}=\adjusted{\transitionModel}_{t}\mean{\stateProcess}_{\initial:t-1},\quad\forall~t\in\timeSpace
\end{equation*}
with the initialization $\mean{\stateProcess}_{\initial}=\initialMean$ given.
The mean of observations is then given by $\mean{\observationProcess}_t\defined\Expectation{\observationProcess_t} = \observationModel_t \mean{\stateProcess}_t$ for all $t\in\timeSpace_-$.
Define the pairwise joint covariances $\covariance_{\stateProcess_{t_1},\stateProcess_{t_2}} \defined \Covariance{\stateProcess_{t_1},\stateProcess_{t_2}}$, $\covariance_{\observationProcess_{t_1},\observationProcess_{t_2}} \defined \Covariance{\observationProcess_{t_1},\observationProcess_{t_2}}$, and $\covariance_{\observationProcess_{t_1},\stateProcess_{t_2}} \defined \Covariance{\observationProcess_{t_1},\stateProcess_{t_2}}$ for $t_1,t_2\in\timeSpace_-$.
Using the state and observation models \eqref{eq:model}, the covariances evolve as
\begin{align*}
  \covariance_{\stateProcess_{t+1},\stateProcess_{t}}&=\adjusted{\transitionModel}_{t+1} \covariance_{\stateProcess_{\initial:t},\stateProcess_{t}}\\
  \covariance_{\stateProcess_{t+1},\stateProcess_{t+1}}&=\adjusted{\transitionModel}_{t+1}\covariance_{\stateProcess_{\initial:t},\stateProcess_{\initial:t}}\adjusted{\transitionModel}_{t+1}^\transpose+\covarianceProcessNoise_{t+1}\\
  \covariance_{\observationProcess_{t_1},\stateProcess_{t_2}}&=\observationModel_{t_1} \covariance_{\stateProcess_{t_1},\stateProcess_{t_2}}\\
  \covariance_{\observationProcess_{t_1},\observationProcess_{t_2}}&=\observationModel_{t_1} \covariance_{\stateProcess_{t_1},\stateProcess_{t_2}} \observationModel_{t_2}^\transpose + \covarianceObservationNoise_{t_1}\indicator{t_1=t_2}
\end{align*}
with initialization $\covariance_{\stateProcess_{\initial},\stateProcess_{\initial}}=\initialCovariance$.
Here $\indicator{t_1=t_2}$ is the indicator function that equals $1$ if $t_1=t_2$ and equals $0$ otherwise.
Stacking all times up to $\horizon$, the above recursions are used to compute the joint covariances \(\covState\), \(\covStateObs\) and \(\covObs\) which are used in the batch smoothing and the causal Wiener-Hopf filter as described in \Sec{sec:classical-methods}.

\addtolength{\textheight}{-12cm}


\bibliographystyle{IEEEtran}
\bibliography{references}

\end{document}